\documentclass[conference,a4paper]{IEEEtran}
\IEEEoverridecommandlockouts

\usepackage{cite}
\usepackage{amsmath,amssymb,amsfonts, amsthm}
\usepackage{algorithmic}
\usepackage{graphicx}
\usepackage{textcomp}
\usepackage{xcolor}
\usepackage{tikz}
\usetikzlibrary{arrows.meta,calc}

\newtheorem{thm}{Theorem}
\newtheorem{defn}{Definition}
\newtheorem{rem}{Remark}
\newtheorem{lem}{Lemma}
\newtheorem{cor}{Corollary}

\def\BibTeX{{\rm B\kern-.05em{\sc i\kern-.025em b}\kern-.08em
    T\kern-.1667em\lower.7ex\hbox{E}\kern-.125emX}}
\begin{document}

\title{Polymatroidal Representations of Aggregate EV Flexibility
Considering Network Constraints
}
\author{\IEEEauthorblockN{Karan Mukhi\IEEEauthorrefmark{1},
Alessandro Abate\IEEEauthorrefmark{1}}
\IEEEauthorblockA{\IEEEauthorrefmark{1}Department of Computer Science, University of Oxford, United Kingdom\\
Email: \{karan.mukhi, alessandro.abate\}@cs.ox.ac.uk}%
}

\IEEEpubid{\makebox[\columnwidth]{979-8-3315-2503-3/25/\$31.00 ©2025 Crown\hfill} \hspace{\columnsep}\makebox[\columnwidth]{}}

\maketitle
\IEEEpubidadjcol

\begin{abstract}
The increasing penetration of electric vehicles (EVs) introduces significant flexibility potential to power systems. However, uncoordinated or synchronous charging can lead to  overloading of distribution networks.
Extending recent approaches that utilize generalized polymatroids, a family of polytopes, to represent the aggregate flexibility of EV populations, we show how to integrate network constraints into this representation to obtain network-constrained aggregate flexibility sets. 
Furthermore, we demonstrate how to optimize over these network-constrained aggregate flexibility sets, and propose a disaggregation procedure that maps an aggregate load profile to individual EV dispatch instructions, while respecting both device-level and network constraints.
\end{abstract}


\section{Introduction}
The rapid roll‐out of electric vehicles (EVs) is transforming electricity demand profiles at the distribution level. Projections for Europe suggest that, by 2035, EVs could account for over 10\% of total electricity consumption, with peak loads that eclipse today's residential evening peaks if charging is left uncoordinated \cite{InternationalEnergyAgency2024Global2024}. While demand from EVs has the potential to overwhelm distribution networks, populations of EVs also represent a vast reservoir of \textit{flexibility}: the ability to shift consumption in time without compromising user needs. Harnessing this flexibility to provide ancillary services is essential to enable power system operators to economically manage the intermittency and stability issues that arise from running grids with high penetrations of renewables. 

To enable their participation in ancillary service markets, a growing body of work has focused on characterizing the set of feasible aggregate load profiles for a population of energy flexible devices. Formally, this involves computing the Minkowski sum of a collection of polytopes \cite{Zhao2017ALoads}. In general computing the exact Minkowski sum is computationally expensive, hence most of these methods tend to focus on computing inner \cite{Taha2024AnPopulations} or outer \cite{Barot2017APolytopes} approximations of the aggregate flexibility  set. Recent work leverage \textit{generalized polymatroids} (g-polymatroids), a rich class of polytopes whose combinatorial structure enables \textit{exact} yet compact representations  of the aggregate flexibility of large, heterogeneous EV fleets \cite{Mukhi2025ExactPolymatroids}. However, most of this work ignores the physical limits of the distribution networks.
In practice, large numbers of EVs could be connected to the same feeders in the network, and hence be exposed to the same network constraints. Failure to consider these limits when co-coordinating large populations of EVs risks overloading the distribution grid \cite{Priyadarshan2025DistributionElectrification}. Hence, in this work we address this issue by extending the g-polymatroid framework to incorporate network constraints.

The contributions of this work are twofold. First, by leveraging a classical result on g-polymatroids, we provide an exact characterization of aggregate flexibility sets that simultaneously capture both device-level and network constraints. Second, we present a disaggregation procedure that, given an optimal aggregate load profile, computes dispatch instructions for individual EVs in the population. The rest of the paper is structured as follows, in Section \ref{sec:unconstrained} we introduce our EV charging model and review previous results detailing how to characterize the aggregate flexibility sets using g-polymatroids. Section \ref{sec:constrained} introduces some theoretical results that extends this formulation to the network constrained case. 
We discuss how to optimize over the proposed aggregate flexibility sets, and present a disaggregation procedure in Section~\ref{sec:disag}.  Finally, in Section \ref{sec:case} we provide a case study to show how this work may be used in practice, and conclusions are drawn in Section \ref{sec:conc}.
\begin{figure}[t]
    \centering
    \includegraphics[width=\columnwidth]{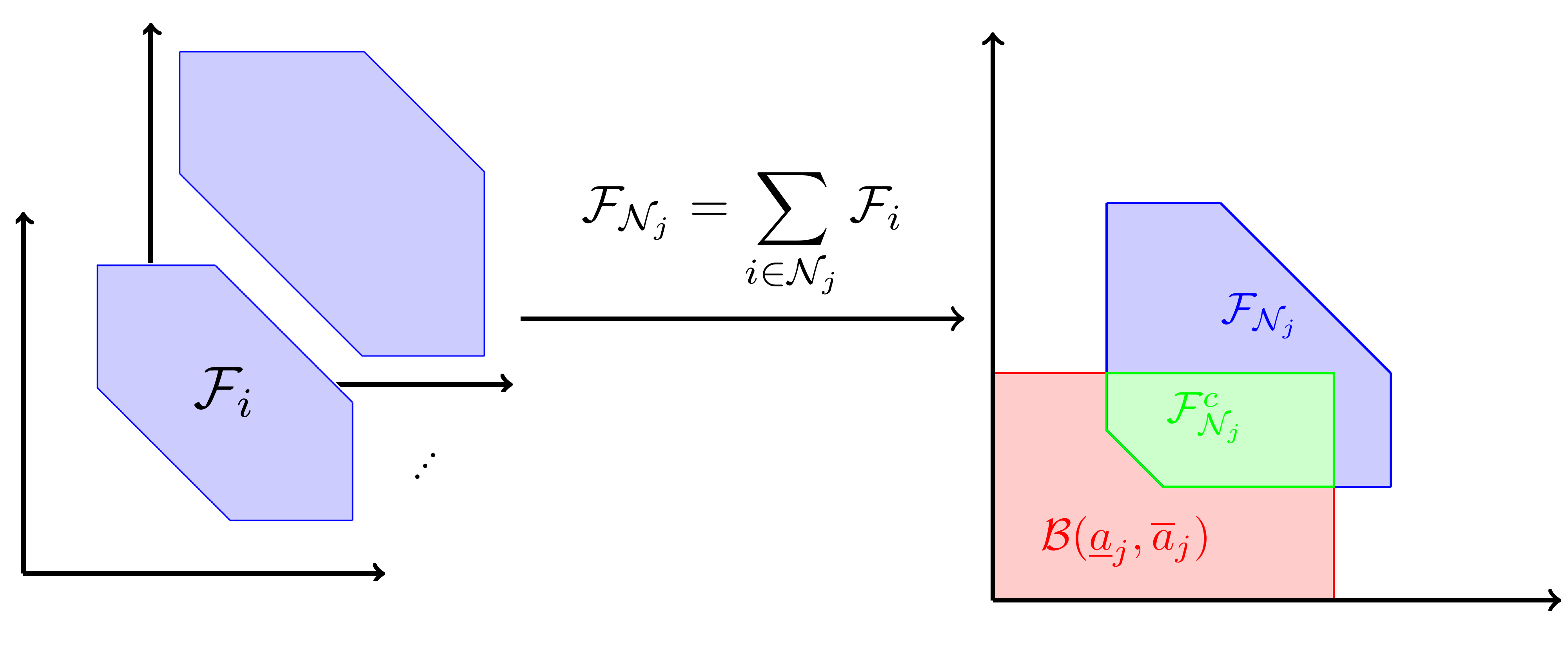}
    \caption{A schematic of the work presented in this paper. We are provided with a set of EVs with operational requirements connected to various nodes, $j \in \mathcal{J}$, on the network. The set of unconstrained aggregate load profiles at node $j$, denoted $\mathcal{F}_{\mathcal{N}_j}$ is the Minkowski sum of the devices connected to it. At each node there will be a set of network constraints $\mathcal{B}(\underline{a}_j, \overline{a}_j)$ (red region). The constrained aggregate flexibility set at node $j$ is defined in \eqref{eq:intersection} as the intersection of $\mathcal{B}(\underline{a}_j, \overline{a}_j)$ and $\mathcal{F}_{\mathcal{N}_j}$, and is denoted $\mathcal{F}_{\mathcal{N}_j}^c$ (green region).}
    \label{fig:schema}
\end{figure}

 \textit{Notation:} For a vector $ x \in \mathbb{R}^{\mathcal{T}} $, where $ \mathcal{T} \subseteq \mathbb{N} $ is a finite index set and $ t \in \mathcal{T} $, we denote by $ x(t) $ the $ t^{\text{th}} $ component of $ x $. For any subset $ A \subseteq \mathcal{T} $, we define $ x(A) := \sum_{t \in A} x(t) $. Lastly, we use the notation $ \sum(\cdot) $ to denote Minkowski summation.

\section{Unconstrained Aggregated Flexibility}\label{sec:unconstrained}
In this section, we present the charging model for a population of EVs and formally define the notion of individual and aggregate flexibility sets. We then review existing results that enable the representation of aggregate flexibility as a generalized polymatroid.

We consider an aggregator that centrally manages the charging of a population of EVs. The EVs are connected to the different feeders in the distribution network, which are indexed by $ j \in \mathcal{J} $. For each node $ j $, the set of EVs connected to that node is indexed by $ i \in \mathcal{N}_j $. The control problem is formulated over a discrete time horizon consisting of $ T $ time steps, indexed by $ t \in \mathcal{T} := \{1, \dots, T\} $.

\subsection{EV Charging Model}
Let $ u_i \in \mathbb{R}^\mathcal{T} $ denote the piecewise constant \textit{charging profile} of EV $i$, where $ u_i(t) $ represents the power consumption in the discrete time interval $t \in \mathcal{T}$. We let $ x_i \in \mathbb{R}^\mathcal{T}$ denote the \textit{state of charge} (SoC) of the EV, where $ x_i(t) $ denotes the SoC at time step $ t $. The evolution of the SoC is governed by the following discrete-time charging dynamics:
\begin{equation}\label{eq:ev_dynamics}
    x_i(t+1) = x_i(t) + u_i(t) \delta,
\end{equation}
where $ \delta $ denotes the duration of each time step. Without loss of generality, we assume $ \delta = 1 $ throughout this work. Furthermore, we assume that the initial state of charge is zero, i.e., $ x_i(0) = 0 $. This assumption also holds without loss of generality, as any non-zero initial SoC can be handled by accordingly shifting the constraints on the SoC. 

We restrict our formulation to a charging-only regime, i.e., vehicle-to-grid (V2G) capabilities are not considered. 
Let $ m_i > 0 $ denote the maximum charging rate of the EV. Each EV is assumed to be connected to the charging infrastructure for a subset of the time horizon. Specifically, let $ \underline{t}_i, \overline{t}_i \in \mathcal{T} $ denote the arrival and departure times, respectively, where $\underline{t}_i < \overline{t}_i$. We denote the subset of EVs connected to node $j$ with homogeneous arrival and departure times as $\mathcal{S}_j^{\underline{t},\overline{t}} \subseteq \mathcal{N}_j$. The EV’s \textit{connection time} is defined as the discrete interval $ \mathcal{C}_i := \{\underline{t}_i,..., \overline{t}_i \} \subset \mathcal{T} $, representing the set of time steps during which the EV is available for charging. Based on this availability, the charging power is subject to the following constraints:
\begin{subequations}\label{eq:power_constraints}
    \begin{align}
        0 \leq u_i(t) &\leq m_i  &&\forall t \in \mathcal{C}_i, \\
              u_i(t) &= 0      &&\forall t \in \mathcal{T} \setminus \mathcal{C}_i. \label{eq:p_cons2}
    \end{align}
\end{subequations}
Lastly, we require that the EVs is fully charged at its departure time $\overline{t}_i$, and so its SoC lies within the range: $\underline{e}_i \leq x_i(\overline{t}_i) \leq \overline{e}_i$, where $\underline{e}_i$ and $\overline{e}_i$ denote the lower and upper bounds, respectively, on the SoC at departure. Alternatively applying the charging dynamics from \eqref{eq:ev_dynamics} and the notation introduced above we have:
\begin{equation}\label{eq:energy_constraints}
    \underline{e}_i \leq u_i(\mathcal{T}) \leq \overline{e}_i.
\end{equation}
 Note, the power constraint of \eqref{eq:p_cons2} implicitly enforces that the EV will be charged when it departs.
\subsection{Flexibility Sets}
We now consider the set of charging profiles a single EV may take and the set of profiles a population of EVs may take. This can be done by considering the set of charging profiles that respect the constraints for each EV.
\begin{defn}
    The individual flexibility set of EV $i$, denoted $\mathcal{F}_i,$ is defined as the set of charging profiles that satisfy all operational constraints of the EV:
    \begin{equation}\label{eq:individual_flex_set}
         \mathcal{F}_i := \left\{ u_i \in \mathbb{R}^T \; \middle\vert \;
        \begin{array}{@{}cl}
                     &  0 \leq u_i(t) \leq 0  \;\;\; \forall t \in \mathcal{T} \setminus \mathcal{C}_i \\
                     &0 \leq  u_i(t) \leq m_i  \;\; \forall t \in \mathcal{C}_i  \\
             &\underline{e}_i \leq u_i(\mathcal{T})\leq \overline{e}_i  \\
        \end{array} 
        \right\}.
    \end{equation}
\end{defn}
The \textit{aggregated charging profile} for the population of EVs connected to node $j$, denoted $u_{\mathcal{N}_j}$, is the sum of the individual charging profiles of all EVs connected to that node. Formally, 
\begin{equation*}
    u_{\mathcal{N}_j} := \sum_{i \in \mathcal{N}_j} u_i.
\end{equation*}

\begin{defn}
    The \emph{aggregate flexibility set} at node $j$, denoted $\mathcal{F}_{\mathcal{N}_j},$is the set of all feasible aggregate charging profiles for the population of EVs connected to node $j$:
    \begin{equation}\label{eq:aggregate_flex_set}
         \mathcal{F}_{\mathcal{N}_j}:= \left\{ u_{\mathcal{N}_j} \in \mathbb{R}^T \mid u_{\mathcal{N}_j} := \sum_{i \in \mathcal{N}_j} u_i,  u_i \in  \mathcal{F}_i\;\; \forall i \in \mathcal{N}_j\right\}.
    \end{equation}
\end{defn}
This is, by definition,  the Minkowski sum of the collection of individual flexibility sets \cite{Barot2017APolytopes}:
\begin{equation}\label{eq:mink_sum}  
     \mathcal{F}_{\mathcal{N}_j} =  \sum_{i \in \mathcal{N}_j} \mathcal{F}_i.
\end{equation}
Representing $\mathcal{F}_{\mathcal{N}_j}$ is challenging, however we can use some tools from combinatorial optimization to help characterize these sets.
\subsection{Generalized Polymatroids}
Generalized polymatroids (g-polymatroids) are a family of polytopes that can be employed to form representations of the Minkowski sum from \eqref{eq:mink_sum}. For brevity, we do not provide a comprehensive overview of generalized polymatroids. Instead, we refer the reader to \cite{Frank2011ConnectionsOptimization} for a rigorous treatment of the subject, and to \cite{Mukhi2025ExactPolymatroids} for a discussion in the context of aggregate flexibility.

A g-polymatroid, denoted $\mathcal{Q}(p,b)$, is defined by a pair of set functions: a supermodular function $ p: 2^\mathcal{T} \rightarrow \mathbb{R} $ and a submodular function $ b: 2^\mathcal{T} \rightarrow \mathbb{R} $, which jointly characterize the lower and upper bounds of the g-polymatroid:
\begin{equation*}
    \mathcal{Q}(p,b) := \{u \in \mathbb{R}^T\; | \; p(A) \leq u(A) \leq b(A) \; \forall A \subseteq \mathcal{T}\}.
\end{equation*}
In previous work, it was shown that the individual flexibility sets defined in \eqref{eq:individual_flex_set} are g-polymatroids, as formally stated in the following lemma. Intuitively, $A$ is a subset of the time steps within the horizon, and $p(A)$ and $b(A)$ denote device's minimum and maximum energy consumption, respectively, during that subset of the time horizon.

\begin{lem}\label{lem:is_poly}
$\mathcal{F}_i$ is the g-polymatroid $Q(p_i, b_i)$, where:
\begin{align*}
    &p_i(A) := \sum_{t}^{|A \cap \mathcal{C}_i|} \underline{v}_i(t), \quad b_i(A) :=  \sum_{t}^{|A \cap \mathcal{C}_i|} \overline{v}_i(t), \\
\end{align*}
where we define $\underline{v}_i, \overline{v}_i \in \mathbb{R}^{\mathcal{C}_i}$:
\begin{equation*}
        \underline{v}_i(t) = 
        \begin{cases}
            0                               & t <  \lfloor \underline{e}_i / m_i \rfloor \\
             \textrm{rem}(\underline{e}_i, m_i )          & t =  \lfloor \underline{e}_i / m_i \rfloor \\
            m_i                              & t \geq  \lfloor \underline{e}_i / m_i \rfloor \\
        \end{cases}
    \end{equation*}
\begin{equation*}
        \overline{v}_i(t) = 
        \begin{cases}
            m_i                              & t <  \lfloor \overline{e}_i / m_i \rfloor \\
            \textrm{rem}(\overline{e}_i, m_i )         & t =  \lfloor \overline{e}_i / m_i \rfloor \\
            0                               & t \geq  \lfloor \overline{e}_i / m_i \rfloor. \\
        \end{cases}
    \end{equation*}
\end{lem}
\begin{proof}
    This follows from \cite[Corollary 3.7]{Mukhi2024DistributionallyFlexibility} and the definitions of $\underline{v}_i$ and $\overline{v}_i$.
\end{proof}
A particularly useful property of g-polymatroids, in this context, is that their Minkowski sum can be computed efficiently.
\begin{thm}\label{thm:g_polymatroid_sum}\cite[Theorem 14.2.15]{Frank2011ConnectionsOptimization}
    The Minkowski sum of a set of g-polymatroids is given by
    \begin{equation*}
        Q\left(p_{\mathcal{N}}, p_{\mathcal{N}}\right) =  \sum_{i \in \mathcal{N}} Q(p_i, b_i)
    \end{equation*}
    where $
        p_{\mathcal{N}} = \sum_{i \in \mathcal{N}} p_i$ and $b_{\mathcal{N}} = \sum_{i \in \mathcal{N}} b_i$.
\end{thm}
Using this theorem and the characterization of the individual flexibility sets as g-polymatroids from Lemma \ref{lem:is_poly}, one can provide an exact representation of the aggregate flexibility sets from \eqref{eq:aggregate_flex_set}.
\begin{cor} The unconstrained aggregate flexibility set, $\mathcal{F}_{\mathcal{N}_j}$, is the g-polymatroid $ \mathcal{Q }\left(p_{\mathcal{N}_j}, b_{\mathcal{N}_j}\right)$ where 
    \begin{equation*}
        p_{\mathcal{N}_j} = \sum_{i \in \mathcal{N}_j} p_i, \quad \quad b_{\mathcal{N}_j} = \sum_{i \in \mathcal{N}_j} b_i.
    \end{equation*}
\end{cor}
This provides us with a succinct representation of the aggregate flexibility in a population of EVs. Note that the flexibility sets for a more general class of devices can also be represented as g-polymatroids, including devices in which discharging is allowed, i.e. EVs \textit{with} V2G capabilities. However, the super- and submodular functions associated with these devices and hence the aggregate sets can be more complex, and so in this paper we will restrict our attention to populations of EVs with charging only capabilities. 

\section{Constrained Aggregated Flexibility}\label{sec:constrained}
In this section we consider the case where the population share a feeder, subjecting the aggregation to a common network constraint, and we modify the super- and submodular functions that generate the aggregations so that they respect these constraints. We start by formally defining the feasible power flows through the feeder and then use results in the literature to show how one can derive the super- and submodular functions generating the network constrained aggregate flexibility sets. Finally, we demonstrate how these sets can be composed to model flexibility across different regions of the network.

\subsection{Network Constraint Sets}
To start we consider the set of feasible power profiles through node $j$ of the network. We let $\underline{r}_j$ and $\overline{r}_j$ denote the lower and upper limits, respectively, on the power flow through the node. We assume there is a nominal, time-dependent, power flow through the node, due to power consumption from non-flexible devices, we denote this with $u^0_j\in \mathbb{R}^\mathcal{T}$. From this we can define our time-dependent lower and upper limits on power flow as $\underline{a}_j(t) = u^0_j(t) - \underline{r}_j$ and $\overline{a}_j(t) =  \overline{r}_j - u^0_j(t) $, respectively. 
\begin{defn}\cite[14.1]{Frank2011ConnectionsOptimization}
    For $\underline{a}, \overline{a} \in \mathbb{R}^T$, where $\underline{a} \leq \overline{a}$ we define a \emph{box} as:
    \begin{equation*}
    \mathcal{B}(\underline{a}, \overline{a}) := \{u \in \mathbb{R}^T\; | \; \underline{a}(t) \leq u(t) \leq \overline{a}(t) \; \forall t  \in  \mathcal{T}\}.
\end{equation*}
\end{defn}
For a feeder with time dependent lower and upper power flow limits, $\underline{a}$ and $\overline{a}$, the set $\mathcal{B}(\underline{a}, \overline{a})$ represents the set of all feasible power flow profiles over the time horizon, through the feeder. Geometrically, this corresponds to the red region in Fig.~\ref{fig:schema}.
\begin{lem}
    The constrained aggregate flexibility set at node $ j $, denoted $ \mathcal{F}^c_{\mathcal{N}_j} $, is defined as:
    \begin{equation}\label{eq:intersection}
        \mathcal{F}^c_{\mathcal{N}_j} = \mathcal{Q}\left(p_{\mathcal{N}_j}, b_{\mathcal{N}_j}\right) \cap \mathcal{B}(\underline{a}_j, \overline{a}_j),
    \end{equation}
\end{lem}
\begin{proof}
    The result follows directly by applying the local power flow constraints at node $ j $, represented by $ \mathcal{B}(\underline{a}_j, \overline{a}_j) $, to the aggregate flexibility set $ \mathcal{F}_{\mathcal{N}_j} = \mathcal{Q}\left(p_{\mathcal{N}_j}, b_{\mathcal{N}_j}\right) $.
\end{proof}
We now aim to derive a $ g $-polymatroid representation of the set $ \mathcal{F}^c_{\mathcal{N}_j} $. In particular, we seek to identify the supermodular and submodular functions that generate the corresponding $ g $-polymatroid structure. To this end, we make use of the following theorem.
\begin{thm}\cite[Theorem 14.3.9]{Frank2011ConnectionsOptimization}\label{thm:box_g_poly}
    Given a g-polymatroid $\mathcal{Q}\left(p,b\right)$, the intersection $\mathcal{F} = \mathcal{Q}\left(p,b\right) \cap \mathcal{B}(\underline{a}, \overline{a})$ is given by $\mathcal{F} = \mathcal{Q}\left( p',b' \right)$ where 
    \begin{align*}
        &p'(A) := \max_{X \subseteq \mathcal{T}} \;\{ p(X) - \underline{a}(X-A) + \overline{a}(A-X)\}, \\
        &b'(A) := \min_{X \subseteq \mathcal{T}} \;\{ b(X) - \overline{a}(X-A) + \underline{a}(A-X)\}.
    \end{align*}
\end{thm}
Using this theorem the super- and submodular functions generating the g-polymatroid representation of the constrained aggregated flexibility set at node $j$, $ \mathcal{F}^c_{\mathcal{N}_j} $  given by 
\begin{align*}
    & p_{\mathcal{N}_j}^c(A) := \max_{X \subseteq \mathcal{T}} \;\{  p_{\mathcal{N}_j}(X) - \underline{a}_j(X-A) + \overline{a}_j(A-X)\}, \\
    & b_{\mathcal{N}_j}^c(A) := \min_{X \subseteq \mathcal{T}} \;\{  b_{\mathcal{N}_j}(X) - \overline{a}_j(X-A) + \underline{a}_j(A-X)\}.
\end{align*}
Essentially this involves solving a submodular function minimization (and supermodular maximization) problem. Details of algorithms that solve the problem in strongly polynomial time for general functions can be found in \cite{Schrijver2000ATime} and \cite{Iwata2001AFunctions}.

\subsection{Composing Aggregations}
We consider an aggregator managing populations of EVs distributed across various feeders within a power distribution network, where each node is subject to its own local power constraints. We seek to characterize the set of aggregate load profiles that can be feasibly consumed by the entire population, while respecting the individual and network-level operational constraints. By using Theorem~\ref{thm:g_polymatroid_sum}, the aggregate flexibility set of devices connected throughout the network can be characterized as follows:
\begin{equation*}
    \mathcal{F}^c_\mathcal{N} = \sum_{j \in \mathcal{J}} \mathcal{F}^c_{\mathcal{N}_j} = \mathcal{Q}(p_\mathcal{N}^c, b_\mathcal{N}^c)
\end{equation*}
where the super and submodular functions are given by
\begin{equation*}
p_\mathcal{N}^c = \sum_{j \in \mathcal{J}}  p_{\mathcal{N}_j}^c, \quad \quad b_\mathcal{N}^c = \sum_{j \in \mathcal{J}}  b_{\mathcal{N}_j}^c,
\end{equation*}
$\mathcal{J}$ denotes the set of feeders, and $ p_{\mathcal{N}_j}^c$, $ b_{\mathcal{N}_j}^c$ are the network-constrained super- and submodular functions, respectively, that characterize the aggregate flexibility at each feeder.
\begin{rem}
In scenarios where device populations are subject to multiple layers of constraints—such as those imposed by the hierarchical structure of a radial distribution network—aggregation can be performed recursively. Specifically, by aggregating the flexibility sets of child nodes connected to a common parent node and subsequently applying Theorem~\ref{thm:box_g_poly}, one can impose the network constraints at the parent node. This process can be repeated hierarchically to capture the full network structure.
\end{rem}

\section{Disaggregation}\label{sec:disag}
Having derived an exact representation of the aggregate flexibility set for a population of network-constrained EVs in the preceding section, in this section we turn our attention to optimizing over this set and disaggregating this optimal aggregate consumption profile. Specifically, we aim to determine an optimal aggregate charging profile that maximizes a system-level objective while satisfying network and device-level constraints. Following this, the optimal aggregate solution must be \textit{disaggregated} across individual EVs within the population, such that the resulting individual charging trajectories collectively realize the optimal aggregate profile while adhering to local constraints.

We consider an aggregator that is tasked with optimizing and disaggregating the consumption of a fleet of EVs connected to the same feeder, in the day-ahead electricity market, subject to network constraints, i.e. optimizing and disaggregating over $\mathcal{F}_{\mathcal{N}_j}^c$ This can be formalized with the following LP:
\begin{equation}\label{eq:lp_g_polymatroid}
\underset{}{\text{minimize}}\;\;  c^Tu \quad  
\text{subject to} \;\; u \in \mathcal{F}_{\mathcal{N}_j}^c,
\end{equation}
where $c \in \mathbb{R}^{\mathcal{T}}$ denotes the day-ahead electricity price vector. Since $\mathcal{F}_{\mathcal{N}_j}^c$ is a g-polymatroid, Problem \eqref{eq:lp_g_polymatroid} admits an efficient solution via the greedy algorithm for g-polymatroids, as established in~\cite[Theorem 14.5.2]{Frank2011ConnectionsOptimization}. The solution requires $T+1$ evaluations of either $p_{\mathcal{N}_j}^c$ or $b_{\mathcal{N}_j}^c$, which can be parallelized. We denote the resulting solution by $u_{\mathcal{N}_j}^*$.

Now, given this optimal aggregate load profile $u_{\mathcal{N}_j}^* \in \mathcal{F}_{\mathcal{N}_j}^c$, we focus on disaggregating this load profile among EVs in the population. Formally, we seek to determine individual load profiles that are feasible for each device, collectively satisfy the coupling network constraints, and sum to the optimal aggregate load profile. That is, for all $i \in \mathcal{N}_j$, we aim to find the set of $u^*_i$ such that
\begin{equation}\label{eq:disaggregation_constraints}
        u^*_{\mathcal{N}_j} = \sum_{i \in \mathcal{N}_j} u_i^* \quad \textrm{s.t.} \quad u_i^* \in \mathcal{F}_i \;\; \forall i
\end{equation}

To accomplish this, we first apply Frank–Wolfe decomposition to represent the optimal aggregate profile $u^*_{\mathcal{N}_j} $ as a convex combination of a subset of the vertices of the \textit{unconstrainted} aggregate flexibility set $\mathcal{F}_{\mathcal{N}_j}$~\cite{Frank1956AnProgramming}. By Carathéodory’s Theorem, any point in a $T$-dimensional polytope can be expressed as a convex combination of at most $T+1$ of its vertices~\cite[Proposition 1.15]{Ziegler2012LecturesPolytopes}. Therefore, we can write:
\begin{equation}\label{eq:convex_combination}
    u^*_{\mathcal{N}_j} = \sum_{k=1}^{T+1} \lambda_k v^{(k)}
\end{equation}
where $v^{(k)}$ are vertices of $\mathcal{F}_{\mathcal{N}_j}$, and we have the following constraints $\sum_{k=1}^{T+1} \lambda_k = 1$ and $\lambda_k \geq 0$ for all $k$. As $\mathcal{F}_{\mathcal{N}_j}$ is defined as the Minkowski sum from \eqref{eq:mink_sum}, every vertex $v^{(k)}\in\operatorname{vert}(\mathcal{F})$ can be written as the sum of vertices of the summand polytopes \cite{Fukuda2004FromPolytopes}; that is, there exist $v_i^{(k)}\in\operatorname{vert}(\mathcal{F}_i)$ such that
\begin{equation}\label{eq:vertex_disagg}
    v^{(k)}  = \sum_{i \in \mathcal{N}_j} v^{(k)}_i,
\end{equation}
where $v^{(k)}$ and $v^{(k)}_i$ are minimzers of the same linear cost function, over $\mathcal{F}_{\mathcal{N}_j}$ and $\mathcal{F}_i$ respectively. 
Using this decomposition, the expression in \eqref{eq:convex_combination} can be rewritten as
$u_{\mathcal{N}_j}^* = \sum_{i \in \mathcal{N}_j}  u^*_i$, where we define 
\begin{equation}\label{eq:disagg}
    u^*_i = \sum_{k=1}^{T+1} \lambda_k v^{(k)}_i.
\end{equation}
By construction, the disaggregation defined in \eqref{eq:disagg} ensures that the disaggregated profiles $ u_i^* $ collectively satisfy the constraints of \eqref{eq:disaggregation_constraints}. Moreover, since each $ u_i^* $ is a convex combination of vertices of $\mathcal{F}_i$, it follows directly that $ u_i^* \in \mathcal{F}_i$ for all $i \in \mathcal{N}$. Thus, the profiles $ \{u_i^*\}_{i \in \mathcal{N}} $ constitute a valid disaggregation of the optimal aggregate solution.

\section{Case Study}\label{sec:case}
\begin{figure}[t]
    \centering
    \includegraphics[width=\columnwidth]{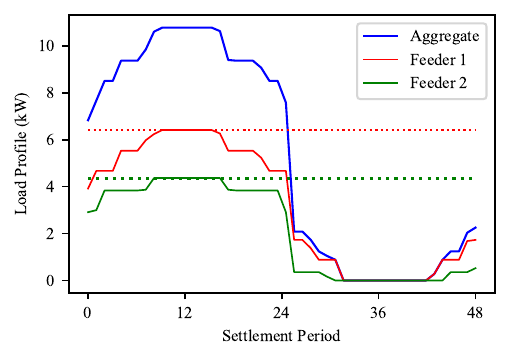}
    \caption{Optimized aggregate load profile for the entire EV population and for each feeder, shown alongside their respective network constraints.}
    \label{fig:aggregated}
\end{figure}
In this section, we present a simple case study to demonstrate the practical application of the proposed framework. We consider an aggregator managing a population of EVs partitioned into two subsets, each connected to the distribution network via a distinct feeder.
The aggregator seeks to minimize the aggregate charging cost of the population in the day-ahead market, over the network-constrained aggregate flexibility set, and disaggregating the aggregate profile across individual devices whilst respecting their local constraints.

We consider this problem over a time horizon consisting of 48 settlement periods. We simulate a population of 10 EVs per feeder group, with operational constraints sampled following the approach described in \cite{Mukhi2025ExactPolymatroids}. This work is scalable to significantly larger EV populations, however we restrict the simulation to 20 vehicles for clarity.
Each feeder is subject to a network constraint, which for clarity we assume is time-invariant. However, in general, the constraints may vary over time. We solve an LP over the aggregate flexibility with a synthesized cost vector.
Fig.~\ref{fig:aggregated} depicts the optimized aggregate consumption profile, along with the aggregate profiles for each feeder. The network constraints imposed at each feeder are also shown for reference. Fig.~\ref{fig:disaggregated} illustrates the aggregate consumption profile at one of the feeders, along with the disaggregated charging profiles of the individual EVs connected to that feeder. 
This demonstrates how the proposed method successfully schedules individual EVs while collectively satisfying the shared network constraints.

\section{Conclusion}\label{sec:conc}
This paper extends the g-polymatroid framework for aggregating flexibility to capture network constraints. By leveraging a classical intersection theorem for g-polymatroids and boxes, we derived super- and submodular set functions that yield an \textit{exact}, compact representation of the network-constrained aggregate charging flexibility of large, heterogeneous EV fleets. These sets inherit the algorithmic advantages of g-polymatroids: linear objectives can be optimized with a greedy algorithm, and the sets compose naturally across multiple feeders, enabling hierarchical modeling of entire feeders. Using this representation, we proposed a disaggregation procedure which maps the optimal aggregate trajectory to individual EV charging profiles that  satisfy both device and network constraints. 
A numerical case study demonstrated that the approach successfully respects feeder constraints while exploiting price signals, thereby unlocking cost-effective flexibility at scale.
Beyond EVs, this framework is applicable to other flexible devices, offering a unifying, tractable foundation for aggregators seeking to 
integrate large volumes of distributed flexibility while guaranteeing network security.

\begin{figure}[!t]
    \centering
    \includegraphics[width=\columnwidth]{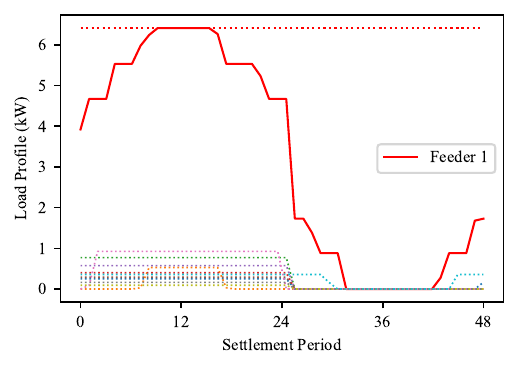}
    \caption{Optimized aggregate load profile at a feeder, along with the disaggregated charging profiles of the individual EVs connected to that feeder.}
    \label{fig:disaggregated}
\end{figure}

\bibliographystyle{IEEEtran}
\bibliography{references}

\end{document}